\documentclass[11pt]{article}

\usepackage{amsthm,amsmath,bm,amssymb,fullpage,times,subfigure}

\newtheorem{theorem}{Theorem}

\newtheorem{definition}{Definition}
\newtheorem{lemma}[theorem]{Lemma}

\newtheorem{corollary}[theorem]{Corollary}

\newcommand{\ket}[1]{\left|#1\right>}
\newcommand{\bra}[1]{\left<#1\right|}

\newcommand{\Z}{{\mathbb{Z}}}
\newcommand{\F}{{\mathbb{F}}}
\newcommand{\E}{{\mathbb{E}}}
\newcommand{\R}{{\mathbb{R}}}

\newcommand{\GL}[0]{{\rm GL}}
\newcommand{\onemat}{{\mathbf{1}}}
\newcommand{\zeromat}{{\mathbf{0}}}
\newcommand{\nix}[1]{{}}

\begin{document}

\title{Quantum algorithms for highly non-linear Boolean functions}

\author{Martin R{\"o}tteler\\
NEC Laboratories America\\
4 Independence Way, Suite 200\\
Princeton, NJ 08540, U.S.A.\\
{\tt mroetteler@nec-labs.com}
}

\maketitle

\begin{abstract} 
 Attempts to separate the power of classical and quantum models of
  computation have a long history. The ultimate goal is to find
  exponential separations for computational problems. However, such
  separations do not come a dime a dozen: while there were some early
  successes in the form of hidden subgroup problems for abelian
  groups--which generalize Shor's factoring algorithm perhaps most
  faithfully--only for a handful of non-abelian groups efficient
  quantum algorithms were found.  Recently, problems have gotten
  increased attention that seek to identify hidden sub-structures of
  other combinatorial and algebraic objects besides groups. In this
  paper we provide new examples for exponential separations by
  considering hidden shift problems that are defined for several
  classes of highly non-linear Boolean functions. These so-called bent
  functions arise in cryptography, where their property of having
  perfectly flat Fourier spectra on the Boolean hypercube gives them
  resilience against certain types of attack. We present new quantum
  algorithms that solve the hidden shift problems for several
  well-known classes of bent functions in polynomial time and with a
  constant number of queries, while the classical query complexity is
  shown to be exponential. Our approach uses a technique that exploits
  the duality between bent functions and their Fourier transforms.
\end{abstract}

\section{Introduction}
A salient feature of quantum computers is that they allow to solve
certain problems much more efficiently than any classical machine.
The ultimate goal of quantum computing is to find problems for which
an exponential separations between quantum and classical models of
computation can be shown in terms of the required resources such as
time, space, communication, or queries. It turns out that the question
about a provably exponential advantage of a quantum computer over
classical computers is a challenging one and examples showing a
separation are not easy to come by. Currently, only few (promise)
problems giving an exponential separation between quantum and
classical computing are known. A common feature they share is that,
simply put, they all ask to extract hidden features of certain
algebraic structures.  Examples for this are hidden shift problems
\cite{vDHI:2003}, hidden non-linear structures \cite{CSV:2007}, and
hidden subgroup problems (HSPs). The latter class of hidden subgroup
problems was studied quite extensively over the past decade.  There
are some successes such as the efficient solution of the HSP for any
abelian group \cite{Shor:97,Kitaev:97}, including factoring and
discrete log as well as Pell's equation \cite{Hallgren:2002}, and
efficient solutions for some non-abelian groups
\cite{FIMSS:2003,BCvD:2005}. However, meanwhile some limitations of
the known approaches to this problem are known \cite{HMRRS:2006} and
presently it is unclear whether the HSP can lend itself to a solution
to other interesting problems such as the graph isomorphism problem.

Most of these methods invoke Fourier analysis over a finite group $G$.
In some sense the Fourier transform is good at capturing some
non-trivial global properties of a function $f$ which at the same time
are hard to figure out for the classical computer which can probe the
function only locally at polynomially many places. For many groups $G$
the quantum computer has the unique ability to compute a Fourier
transform for $G$ very efficiently, i.\,e., in time ${\log^{O(1)} n}$,
where $n$ is the input size. Even though the access to the Fourier
spectrum is somewhat limited, namely via sampling, it nevertheless has
been shown that this limited access can be quite powerful.
Historically, the first promise problems which tried to leverage this
power were defined for certain classes of Boolean functions: the
Deutsch-Jozsa problem \cite{DJ:92} is to decide whether a Boolean
function $f: \Z_2^n \rightarrow \Z_2$ that is promised to be either
constant or a balanced function is actually constant or balanced. In
the Fourier picture this asks to distinguish between functions that
have all their spectrum supported on the $0$ frequency and functions
which have no $0$ frequency component at all. It therefore comes as no
surprise that by sampling from the Fourier spectrum the problem can be
solved. Furthermore, it can be shown that any deterministic classical
algorithm must make an exponential number of queries.  However, this
problem can be solved on a bounded error polynomial time classical
machine. Hence other, more challenging, problems were sought which
asked for more sophisticated features of the function $f$ and were
still amenable to Fourier sampling. One such problem is to identify
$r\in \Z_2^n$ from black box access to a linear Boolean function $f(x)
= r x$, where $x\in \Z_2^n$. Again, in the Fourier domain the picture
looks very simple as each $f$ corresponds to a perfect delta peak
localized at frequency $r$, leading to an exact quantum algorithm
which identifies $r$ using a single query.  Classically, it can be
shown that $\Theta(n)$ queries are necessary and sufficient to
identify $r$ with bounded error.  Based on the observation that a
quantum computer can even handle the case well in which access to $x$
is not immediate but rather through solving another problem of a
smaller size, Bernstein and Vazirani \cite{BV:97} defined the
recursive Fourier sampling (RFS) problem by organizing many instances
of learning a hidden linear function in a tree-like fashion. By
choosing the height of this tree to be $\log n$ they showed a
separation between quantum computers, which can solve the problem in
$n$ queries, and classical computers which require $n^{\log n}$
queries. Soon after this, more algorithms were found that used the
power of Fourier sampling over an abelian group, namely Simon's
algorithm \cite{Simon:94} for certain functions $f:\Z_2^n \rightarrow
\Z_2^{n-1}$, and Shor's algorithms \cite{Shor:97}, where $f$ was
defined on cyclic groups and products thereof, eventually leading to
the HSP.

The idea to achieve speedups from Boolean functions themselves has
obtained significantly less attention. Recently, Hallgren and Harrow
\cite{HH:2008} revisited the RFS problem and showed that other unitary
matrices can serve the role of the Fourier transform in the definition
of RFS problems. They have obtained superpolynomial speedups over
classical computing for a wide class of Boolean functions and unitary
matrices, including random unitary matrices. Together with lower bound
results \cite{Aaronson:2003} this gives a reasonably good
understanding of the power and limitations of the RFS problem. In
another important development, it was shown that the ability to
efficiently perform Fourier transforms on a quantum computer can also
be used to efficiently perform correlations between certain functions.
In the so-called hidden shift problem defined by van Dam, Hallgren,
and Ip \cite{vDHI:2003} this was used in the context of computing a
correlation between a black box implementation of
$f(x)=\left(\frac{x+s}{p}\right)$, where $\left(\frac{x}{p}\right)$
denotes the Legendre symbol and $s\in \Z_p$ is a fixed element, and
the Legendre symbol itself. The main idea behind this is that the
Fourier transform of a shifted function picks up a linear phase which
depends on the shift. Since a correlation corresponds to point-wise
multiplication of the Fourier transforms and since the Legendre symbol
is its own Fourier transform, the correlation can be performed by
computing the Legendre symbol into the phase, leading to an efficient
algorithm that needs only a constant number of queries. The classical
query complexity of this problem is polynomial in $\log p$.

{\bf Our results.} Our main contribution is a generalization of the
hidden shift problem for a class of Boolean functions known as bent
functions \cite{Rothaus:76}. Bent functions are those Boolean
functions for which the Hamming distance to the set of all linear
Boolean functions is maximum (based on comparing their truth tables).
For this reason bent functions are also called maximum non-linear
functions.\footnote{Note that high nonlinearity of a function refers
  to the spectral characterization, i.\,e., the Hamming weight of the
  highest non-zero frequency component is high. It does not imply that
  $f(x) = \sum_{\nu \in \Z_2^n} \alpha_\nu x^\nu$, when written as a
  multivariate polynomial over $\F_2$, has a high (algebraic) degree,
  defined as the maximum degree of any monomial $x^\nu$.  Indeed,
  there are many examples of highly nonlinear functions whose
  algebraic degree is $2$.}  A direct consequence of this is that the
Fourier transform of a bent function $f$ is perfectly flat, i.\,e., in
absolute value all Fourier coefficients, which are defined with
respect to the real valued function $x \mapsto (-1)^{f(x)}$, are equal
and as small as possible. This feature of having a flat Fourier
spectrum is desirable for cryptographic purposes because, roughly
speaking, such a function is maximally resistant against attacks that
seek to exploit a dependence of the outputs on some linear subspace of
the inputs. It turns out that bent functions exist if and only if the
number of variables is even and that there are many of them:
asymptotically, the number of bent functions in $n$ variables is at
least $\Omega\left( \left(\frac{2^{n/2+1}}{e}\right)^{2^{n/2}}
  \sqrt{2\pi 2^{n/2}}\right)$, see for instance \cite{CG:2006}. What
is more, several explicit constructions of infinite families of bent
functions are known and they are related to so-called difference sets
which are objects studied in combinatorics.  Since the Fourier
transform of $f$ is flat and the Boolean Fourier transform is real, it
follows that (up to normalization) the Fourier spectrum takes only
values $\pm 1$, i.\,e., it again is described by a Boolean function,
called the {\em dual} bent function and denoted by $\widetilde{f}$.
Arguably, the most prominent example for a bent function is the inner
product function $ ip_n(x_1, \ldots, x_{n}) = \sum_{i=1}^{n/2}
x_{2i-1} x_{2i} $ written in short as $ip_n(x,y)=xy^t$.  This function
can be generalized to $f(x,y) = x\pi(y)^t+g(y)$, where $\pi$ is an
arbitrary permutation of strings of length $n/2$ and $g:\Z_2^{n/2}
\rightarrow \Z_2$ is an arbitrary function. This leads to the class of
so-called Maiorana-McFarland bent functions. The dual bent function is
then given by the Boolean function $\widetilde{f}(x,y) = \pi^{-1}(x) y^t
+ g(\pi^{-1}(x))$.

We define the hidden shift problem for a fixed bent function $f$ as
follows: an oracle ${\cal O}$ provides us with access to $f$ and $g$,
where $g$ is promised to be a shifted version of $f$ with respect to
some unknown shift $s$. Using oracles of this kind, we show an
exponential separation of the quantum and classical query complexity
of the hidden shift problem, the former being at most linear, the
latter being exponential. Furthermore, we also consider a variation
of the problem where an oracle $\widetilde{{\cal O}}$ in addition
provides oracle access to the dual bent function $\widetilde{f}$. We
show that $s$ can be extracted from $\widetilde{{\cal O}}$ by a
quantum algorithm using one query to $f$ and one query to
$\widetilde{f}$. We present two other classes of bent functions,
namely the partial spread class defined by Dillon \cite{Dillon:75} and
a class defined by Dobbertin \cite{Dobbertin:95}, which uses
properties of certain Kloosterman sums over finite fields to show the
bentness of the functions.

What is the significance of our result? In short, we provide new
examples for exponential separations between quantum and classical
computing. The class of problems studied in this paper yields a large
new set of problems for exponential separations in query complexity
with respect to oracles. A feature of the quantum algorithms presented
here are their simplicity in that besides classical computation of
function values the only quantum operation required are the Fourier
transform over the groups $\Z_2^n$. 

How does this relate to other separations? While exponential
separations in query complexity were known before, for instance for
abelian hidden subgroup problems, the hidden shift problems for bent
functions are the first problems for which such a separation can be
shown from {\em Boolean} functions. In the case of abelian HSP for
order $2$ subgroups of $\Z_2^n$, it is possible to assume that the
functions hiding the hidden subgroup take the form $f(x) = \pi(Ax)$,
where $A\in \F_2^{(n-1)\times n}$ is a matrix of rank $n-1$, and $\pi$
is a permutation of strings of length $n-1$. The goal is to find a
vector $s \in \F_2^n$ in the kernel of $A$. Note that these functions
are not Boolean functions but rather functions from $\Z_2^n
\rightarrow \Z_2^{n-1}$. To the best of our knowledge the best
separations that were obtainable so far from Boolean functions were
the superpolynomial separations shown in \cite{HH:2008}. Those were
obtained by generalizing the ideas of recursive Fourier sampling from
parity functions to more general classes of Boolean functions.

{\bf Related work.} The techniques used in this paper are related to
the techniques used in \cite{vDHI:2003}, in particular the method of
using the Fourier transform thrice in order to correlate a shifted
function with a given reference function, thereby solving a
deconvolution problem. We see the main difference in the richness of
the class of Boolean functions for which the method can be applied and
the query lower bound.

It was observed in \cite{FIMSS:2003,Kuperberg:2005} that the hidden
shift problem for {\em injective} functions $f,g: G \rightarrow S$
from an abelian $G$ to a set $S$ is equivalent to hidden subgroup
problem over $G \rtimes \Z_2$, where the action of $\Z_2$ on $G$ is
given by the inverse. There are several other papers that deal with
the injective hidden shift problem over abelian and non-abelian groups
\cite{CvD:2007,CW:2007,MRR+:2007}. In contrast, the functions studied
here are defined on the abelian group $\Z_2^n$ and very far from being
injective. As we show it will be nevertheless possible to define a
related hidden subgroup problem over an elementary abelian group,
however, for this we have to consider ``quantum functions'' to encode
the period.
 
Perhaps most closely related to our scenario is the work by Russell
and Shparlinski \cite{RS:2004} who considered shift problems for the
case of $\chi(f(x))$, where $f$ is a polynomial on a finite group $G$
and $\chi$ a character of $G$, a general setup that includes our
scenario.  The two cases for which algorithms were given in
\cite{RS:2004} are the reconstruction of a monic, square-free
polynomial $f \in \F_p[X]$, where $\chi$ is the quadratic character
(Legendre symbol) over $\F_p$ and the reconstruction of a hidden shift
over a finite group $\chi(sx)$, where $\chi$ is the character of a
known irreducible representation of $G$.  The technique used in
\cite{RS:2004} is a generalization of the technique of
\cite{vDHI:2003}. In the present paper we extend the class of
functions for which the hidden shift problem can be solved to the case
where $f$ is a multivariate polynomial and $G$ is the group $\Z_2^n$.

Related to the hidden shift problem is the problem of unknown shifts,
i.\,e., problems in which we are given a supply of quantum states of
the form $\ket{D+s}$, where $s$ is random, and $D$ has to be
identified. Problems of this kind have been studied by Childs,
Vazirani, and Schulman \cite{CSV:2007}, where $D$ is a sphere of
unknown radius, Decker, Draisma, and Wocjan \cite{DDW:2008}, where $D$
is a graph of a function, and Montanaro \cite{Montanaro:2009}, where
$D$ is the set of points of a fixed Hamming-weight. The latter paper
also considers the cases where $D$ hides other Boolean functions such
as juntas, a problem that was also studied in \cite{AS:2007}.  In
contrast to all these problems in our case the set $D$ is already
known, but the shift $s$ has to be identified.

We are only aware of relatively few occasions where bent functions
have been used in theoretical computer science: they were used in the
context of learning of intersections of halfspaces \cite{KS:2007},
where they gave rise to maximum possible number of slicings of edges
of the hypercube.  Also the recent counterexample for failure of the
inverse Gowers conjecture in small characteristic \cite{LMS:2008} uses
a special bent function.

\section{Fourier analysis of Boolean functions}

We recall some basic facts about Fourier analysis of Boolean
functions, see also the recent review article \cite{DeWolf:2008} for
an introduction. Let $f : \Z_2^n \rightarrow \R$ be a real valued
function on the $n$-dimensional Boolean hypercube. The Fourier
representation of $f$ is defined as follows. First note that for any
subset $S \subseteq [n]=\{1, \ldots, n\}$ we can define a character of
$\Z_2^n$ via $\chi_S: x \mapsto (-1)^{S x^t}$, where $x\in \Z_2^n$
(the transpose is necessary as we assume that all vectors are row
vectors). The inner product of two functions on the hypercube is
defined as $\langle f,g\rangle = \frac{1}{2^n}\sum_x f(x) g(x) =
\E_x(fg)$. The $\chi_S$ are inequivalent characters of $\Z_2^n$, hence
they obey the orthogonality relation $\E_x(\chi_S \chi_T) =
\delta_{S,T}$. The Fourier transform of $f$ is a function
$\widehat{f}: \Z_2^n \rightarrow \R$ defined by
\begin{equation}\label{eq:usualDFT}
\widehat{f}(S) = \E_x(f \chi_S) = \frac{1}{2^n} \sum_{x \in \Z_2^n}
\chi_S(x) f(x),
\end{equation}
$\widehat{f}(S)$ is the Fourier coefficient of $f$ at frequency $S$,
the set of all Fourier coefficients is called the Fourier spectrum of
$f$ and we have the representation $f = \sum_S \widehat{f}(S) \chi_S$.
Two useful facts about the Fourier transform of Boolean functions are
Parseval's identity and the convolution property.  Parseval's identity
says that $\|f\|^2_2 = \sum_S \widehat{f}(S)^2$ which is a special
case of $\langle f,g \rangle = \sum_S \widehat{f}(S) \widehat{g}(S)$.
For two Boolean functions $f,g : \Z_2^n \rightarrow \R$ their
convolution $(f * g)$ is the function defined as $(f*g)(x) =
\frac{1}{2^n} \sum_{y \in \Z_2^n} f(x+y) g(y)$. A standard feature of
the Fourier transform is that it maps the group operation to a point
wise operation in the Fourier domain. Concretely, this means that 
$\widehat{f*g}(S) = \widehat{f}(S) \widehat{g}(S)$, i.\,e.,
convolution becomes point-wise multiplication and vice-versa.

In quantum notation the Fourier transform on the Boolean hypercube
differs slightly in terms of the normalization and is given by the
unitary matrix 
\[
H_{2^n} = \frac{1}{\sqrt{2^n}}\sum_{x,y\in \Z_2^n}
(-1)^{xy^t} \ket{x}\bra{y}.
\]
This is sometimes called the Hadamard transform \cite{NC:2000}. In
this paper we will also use the Fourier spectrum defined with respect
to the Hadamard transform which differs from (\ref{eq:usualDFT}) by a
factor of $2^{-n/2}$. It is immediate from the definition of $H_{2^n}$
that it can be written in terms of a tensor (Kronecker) product of the
Hadamard matrix of size $2\times 2$, namely $H_{2^n} = (H_2)^{\otimes
  n}$, a fact which makes this transform appealing to use on a quantum
computer since it can be computed using $O(n)$ elementary operations.
Also note that in the context of cryptography also the name
Walsh-Hadamard transform for $H_{2^n}$ is common.

Another note on a convention which applies when we consider $\Z_2$
valued functions $f:\Z_2^n \rightarrow \Z_2$. Then we tacitly assume
that the real valued function corresponding to $f$ is actually $F : x
\mapsto (-1)^{f(x)}$.  The Fourier transform is then defined with
respect to $F$, i.\,e., we obtain that
\begin{equation}\label{eq:unitaryDFT}
  \widehat{F}(w) =
  \frac{1}{2^n} \sum_{x \in \Z_2^n} (-1)^{wx^t+f(x)},
\end{equation}
where we use $w\in \Z_2^n$ instead of $S\subseteq [n]$ to denote the
frequencies. Other than this notational convention, the Fourier
transform used in (\ref{eq:unitaryDFT}) for Boolean valued functions
and the Fourier transform used in (\ref{eq:usualDFT}) for real valued
functions are the same.  In the paper we will sloppily identify
$\widehat{f} = \widehat{F}$ and it will be clear from the context
which definition has to be used.

\section{Bent functions}

\begin{definition}
  Let $f : \Z_2^n \rightarrow \Z_2$ be a Boolean function. We say that
  $f$ is {\em bent} if the Fourier coefficients $\widehat{f}(w) =
  \frac{1}{2^n} \sum_{x \in \Z_2^n} (-1)^{wx^t+f(x)}$ satisfy
  $|\widehat{f}(w)| = {2^{-n/2}}$ for all $w \in \Z_2^n$, i.\,e., if
  the spectrum of $f$ is flat.
\end{definition} 
Necessary for bent functions in $n$ variables to exist is that $n$ is
even \cite{Dillon:75,MS:77}. If $f$ is bent, then this implicitly
defines another Boolean function via $2^{n/2} \widehat{f}(w) =:
(-1)^{\widetilde{f}(w)}$. Then this function $\widetilde{f}$ is again
a bent function and called the dual bent function of $f$. By taking
the dual twice we obtain $f$ back: $\widetilde{\widetilde{f}} = f$.

\subsection{A first example: the inner product function}

The most simple bent function is $f(x, y):= x y$ where $x,y \in \Z_2$.
It is easy to verify that $f$ defines a bent function. This can be
generalized to $2n$ variables \cite{MS:77} and we obtain the inner
product
\[
ip_n(x_1, \ldots, x_n, y_1, \ldots, y_n) :=\sum_{i=1}^n x_i y_i.
\] 
Again, it is easy to see that $ip_n$ is bent.
In Section \ref{sec:MMclass} we will see that $ip_n$ belongs to a much
larger class of bent functions. There (in Lemma \ref{lem:MM}) we also
establish that that $ip_n=\widetilde{ip}_n$ is its own dual bent
function which also implies that the vector
$[(-1)^{ip_n(x,y)}]_{x,y\in \Z_2^n}$ is an eigenvector of $H_{2^n}$.
This should be compared to \cite{vDHI:2003} where it was used that the
Legendre symbol $\left(\frac{\cdot}{p}\right)$ gives rise to an
eigenvector of the Fourier transform ${\rm DFT}_p$ over the cyclic
group $\Z_p$. The shift problem for the inner product function is
closely related to the Fourier sampling problem of finding a string
$a$ that is hidden by the function $f(a,x) = ax^t$ \cite{BV:97}, and
indeed the string $a$ can be readily identified from the state
$\frac{1}{\sqrt{2^n}} \sum_{x\in \Z_2^n} (-1)^{ax^t} \ket{x}$.  In the
hidden shift problem the problem is to identify $(a,b)$ from
$\frac{1}{2^n} \sum_{x,y \in \Z_2^n} (-1)^{ip_n(x+a,y+b)}\ket{x,y}$.
This state is up to a global phase given by $\frac{1}{2^n}
\sum_{x,y\in \Z_2^n} (-1)^{xy^t + xb^t + ya^t} \ket{x,y}$.  By
computing $ip_n$ into the phase the latter can be mapped to
$\frac{1}{2^n} \sum_{x,y\in \Z_2^n} (-1)^{xb^t+ya^t} \ket{x,y}$. From
this state the string $(a,b)$ can be extracted by applying to it a
Boolean Fourier transform followed by measurement in the computational
basis.

\subsection{Bent function families}\label{sec:MMclass}

Many examples of bent functions are known and we briefly review some
of these classes. Recall that any quadratic Boolean function $f$ has
the form $f(x_1,\ldots,x_n) = \sum_{i<j} q_{i,j} x_i x_j + \sum_i
\ell_i x_i$ which can be written as $f(x) = x Q x^t + Lx^t$, where
$x=(x_1, \ldots, x_n)\in \Z_2^n$.  Here, $Q\in \F_2^{n\times n}$ is an
upper triangular matrix and $L \in \F_2^{n}$.  Note that since we are
working over the Boolean numbers, we can without loss of generality
assume that the diagonal of $Q$ is zero (otherwise, we can absorb the
terms into $L$). It is useful to consider the associated symplectic
matrix $B = (Q + Q^t)$ with zero diagonal which defines a symplectic
form ${\cal B}(u,v) = u B v^t$.  This form is non-degenerate if and
only if ${\rm rank}(B)=n$.  The coset of $f+R(n,1)$ of the first order
Reed-Muller code is described by the rank of $B$. This follows from
Dickson's theorem \cite{MS:77} which gives a complete classification
of symplectic forms over $\Z_2$:

\begin{theorem}[Dickson]\label{th:Dickson}
  Let $B\in \Z_2^{n \times n}$ be a symmetric matrix with zero
  diagonal (such matrices are also called symplectic matrices). Then
  there exists $R \in {\rm GL}(n, \Z_2)$ and $h\in [n/2]$ such that
  $RBR^t = D$, where $D$ is the matrix $(\onemat_{h} \otimes \sigma_x)
  \oplus \zeromat_{n-2h}$ considered as a matrix over $\Z_2$ (where
  $\sigma_x$ is the permutation matrix corresponding to $(1,2)$). In
  particular, the rank of $B$ is always even. Furthermore, under the
  base change given by $R$ the function $f$ becomes the quadratic form
  $ip_h(x_1, \ldots, x_{2h}) + L^\prime(x_1, \ldots, x_{n})$ where we
  used the inner product function $ip_h$ and a linear function
  $L^\prime$.
\end{theorem}

Next, we give a characterization of the Fourier transform of an affine
transform of a bent function.

\begin{lemma}[Affine transforms]\label{lem:aff}
  Let $f$ be a bent function, let $A\in \GL(n, \Z_2)$ and $b \in
  \Z_2^n$, and define $g(x) := f(xA+b)$. Then also $g(x)$ is a bent
  function and $\widehat{g}(w) = (-1)^{-w(A^{-1})^t b}
  \widehat{f}(w(A^{-1})^t)$ for all $w\in \Z_2^n$.
\end{lemma}

\begin{proof}
We compute $\widehat{g}(w)$ using the substitution $y=xA+b$
as follows:
\begin{eqnarray*}
\widehat{g}(w) & = & \frac{1}{2^n} \sum_x (-1)^{wx^t + f(xA+b)} \\
& = & \frac{1}{2^n} \sum_y (-1)^{w\cdot (A^{-1})^t (y-b)^t + f(y)} \\
& = & \frac{1}{2^n} (-1)^{-w(A^{-1})^t b} \sum_y (-1)^{w(A^{-1})^t y^t + f(y)} \\
& = & (-1)^{-w(A^{-1})^t b} \widehat{f}(w(A^{-1})^t).
\end{eqnarray*}
\vspace*{-0.2cm}
\end{proof}

By combining Theorem \ref{th:Dickson} and Lemma \ref{lem:aff} we
arrive at the following corollary which characterizes the class of
quadratic bent functions.

\begin{corollary}
  Let $f(x) = x Q x^t + Lx^t$ be a quadratic Boolean function such that
  the associated symplectic matrix $B=(Q+Q^t)$ satisfies ${\rm
    rank}(B)=2h=n$. Then $f$ is a bent function. The dual of this bent
  function is again a quadratic bent function.
\end{corollary}

A complete classification of all bent functions has only been achieved
for $n=2, 4$, and $6$ variables. For larger number of variables some
families are known, basically coming from ad hoc constructions. We
present another one of the known families called {\bf M} (Maiorana and
McFarland). First, we remark there are also constructions for making
new bent functions from known ones, the simplest one takes two bent
functions $f$ and $g$ in $n$ and $m$ variables and outputs $(x, y)
\mapsto f(x) \oplus g(y)$. The class {\bf M} of Maiorana-McFarland
bent functions consists of the functions $f(x,y) := x \pi(y)^t+g(y)$,
where $\pi$ is an arbitrary permutation of $\Z_2^n$ and $g$ is an
arbitrary Boolean function depending on $y$ only. The following lemma
characterizes the dual of a bent function in {\bf M}.

\begin{lemma}\label{lem:MM}
  Let $f(x,y) := x \pi(y)^t+g(y)$ be a Maiorana-McFarland bent
  function. Then the dual bent function of $f$ is given by
  $\widetilde{f}(x,y) = \pi^{-1}(x)y^t + g(\pi^{-1}(x))$.
\end{lemma}

\begin{proof}
Let
$\widehat{f}(u,v)$ be the Fourier transform of $f$ at $(u,v)\in
\Z_2^{2n}$.  We obtain
\begin{eqnarray*}
\widehat{f}(u,v) & = & \frac{1}{2^{2n}}\sum_{x,y\in \Z_2^n} 
(-1)^{f(x,y) + (u,v)(x,y)^t} \\
&=& \frac{1}{2^{2n}} \sum_{x,y\in \Z_2^n}(-1)^{x\pi(y)^t+g(y)+(u,v)(x,y)^t}\\
&=& \frac{1}{2^{2n}} \sum_{y\in\Z_2^n}(-1)^{vy^t+g(y)}\left(\!\sum_{x\in\Z_2^n} (-1)^{(u+\pi(y))x^t}\!\right)\\
&=& \frac{1}{2^{n}} \sum_{y\in\Z_2^n}(-1)^{vy^t+g(y)} \delta_{u,\pi(y)} \\
&=& \frac{1}{2^{n}} (-1)^{v\pi^{-1}(u)^t+g(\pi^{-1}(u))}.
\end{eqnarray*}
Hence the dual bent function is given by
$\widetilde{f}(x,y)=\pi^{-1}(x)y^t + g(\pi^{-1}(x))$.
\end{proof}

Another class of bent functions called {\bf PS} (partial spreads) was
introduced by Dillon \cite{Dillon:75} and provides examples of bent
functions outside of {\bf M}.

\begin{theorem} \cite{Dillon:75} \label{partialSpread} Let $U_1,
  \ldots, U_{2^{n/2-1}}$ be $n/2$-dimensional subspaces of $\Z_2^{n}$
  such that $U_i\cap U_j = \{ 0 \}$ holds for all $i\not=j$. Let
  $\chi_i$ be the characteristic function of $U_i$.  Then $f :=
  \sum_{i=1}^{2^{n/2-1}} \chi_i$ is a bent function.
\end{theorem}
A collection of sets $U_i$ as in Theorem \ref{partialSpread} is called
a {\em partial spread}. Explicitly, the $U_i$ can be chosen as $U_i =
\{ (x, a_i x) : x \in \F_{2^{n/2}}\}$ where $a_i \in
\F_{2^{n/2}}^\times$ satisfies $g(a_i)=1$ for a fixed balanced
function $g$. Here we have identified $\Z_2^n$ with the finite field
$\F_{2^n}$ by choosing a polynomial basis.  This provides an explicit
construction for bent functions in ${\bf PS}$. A further class defined
by Dobbertin has the property to include ${\bf M}$ and ${\bf PS}$ is
defined as follows: first, identify $\Z_2^n$ with $\F_{2^{n/2}} \times
\F_{2^{n/2}}$. Let $g$ be a balanced Boolean function of $n/2$
variables, $\varphi$ be a permutation of $\F_{2^{n/2}}$ and $\psi$ be
an arbitrary map from $\F_{2^{n/2}}$ to $\F_{2^{n/2}}$.  Then
\[ 
f(x,y) := \left\{ \begin{array}{c@{\;:\;}l} 
g\left(\frac{x+\psi(\varphi^{-1}(y))}{\varphi^{-1}(y)}\right) 
 & \mbox{if}\; y \not= 0, \\
    0 & \mbox{if}\; y = 0
\end{array}
\right.
\]
is a bent function. 

There are other constructions of bent functions by means of so-called
trace monomials. For this connection, an understanding of certain
Kloosterman sums turns out to be important.  Recall that the
Kloosterman sum in $\F_{2^n}$ is defined as $Kl(a) = \sum_{x \in
  \F_{2^n}^\times} (-1)^{{\rm tr}(x^{-1}+ax)}$, where
$\F_{2^n}^\times$ denotes the non-zero elements of $\F_{2^n}$ and
${\rm tr}$ denotes the trace map from $\F_{2^n}$ to $\Z_2$. For $a \in
\F_{2^n}$ let $f_a(x)$ be the Boolean function $f_a(x) = {\rm tr}(a
x^{2^{n/2-1}})$. It is known that if $a$ is contained in the subfield
$\F_{2^{n/2}}$ and $Kl(a)=-1$, then $f_a$ is a bent function
\cite{Dillon:75}. The existence of such an element $a$ was conjectured
in Dillon's paper and was proved in \cite{LW:90} (see also
\cite{HW:99}) where its existence was shown for all $n$, thereby
showing existence bent functions in this class of trace monomials.

\subsection{Other characterizations of bent functions}

Finally, we note that there are many other characterizations of bent
functions via other combinatorial objects, in particular difference
sets. The connection is rather simple: we get that $D_f := \{x :
f(x)=1\}$ is a difference set in $\Z_2^n$, i.\,e., the set $\Delta D_f
= \{d_1-d_2: d_1,d_2\in D_f\}$ of differences covers each non-zero
element of $\Z_2^n$ an equal number of times.  We briefly highlight
some other connections to combinatorial objects in the following:

\paragraph{Circulant Hadamard matrices.} Bent functions give rise to
  Hadamard matrices of size $2^n \times 2^n$ in a very natural way as
  group circulants as follows. Let $A_f := ((-1)^{f(x+y)})_{x,y \in
    \Z_2^n}$, then $f$ is bent if and only if $A_f$ is a Hada\-mard
  matrix, i.\,e, $A_f A_f^\dagger = n \onemat_n$. Another way of
  saying this is that the shifted functions $x \mapsto (-1)^{f(x+s)}$
  for $s\in \Z_2^n$ are orthogonal. Moreover, in the basis given by
  the columns of $H_{2^n}$ the matrix $A_f$ becomes diagonal, the
  diagonal entries being $\widetilde{f}(x)$.
\paragraph{Balanced derivatives.} Besides the property of $A_f$ being
  a Hadamard matrix another equivalent characterizations of $f$ to be
  bent is that the function $\Delta_h(f) := f(x+h) + f(x)$ is a
  balanced Boolean function (i.\,e., $f$ takes $0$ and $1$ equally
  often) for all non-zero $h$.
\paragraph{Reed-Muller codes.} Bent functions can also be
  characterized in terms of the Reed-Muller codes \cite{MS:77}. Recall
  that the set of all truth tables (evaluations) of all polynomials
  over $\Z_2$ of degree up to $r$ in $n$ variables is called the
  Reed-Muller $R(n,r)$. Then bent functions correspond to functions
  which have the maximum possible distance to all linear functions,
  i.\,e., elements of $R(n,1)$. Quadratic bent functions in $R(n,2)$
  are of particular interest. They correspond to symplectic forms of
  maximal rank and play a role, e.\,g., in the definition of the
  Kerdock codes.
\paragraph{Difference sets.} Finally, we note that bent functions are
  equivalent to objects known as difference sets in combinatorics,
  namely difference sets for the elementary abelian groups $\Z_2^n$
  \cite{BJL1:99}. A difference set is defined as follows: Let $G$ be a
  finite group of order $v=|G|$. A $(v,k,\lambda)$-difference set in
  $G$ is a subset $D \subseteq G$ such that the following properties
  are satisfied: $|D|=k$ and the set $\Delta D = \{a-b : a,b \in D,
  a\not=b \}$ contains every element in $G$ precisely $\lambda$ times.
  Examples for difference sets are for instance the set $D=\{x^2 : x
  \in \F_q\}$ of all squares in a finite field. Here the group $G$ is
  the additive group of $\F_q$, where $q\equiv 3 \, ({\rm mod} \; 4)$
  is a prime power. The parameters of this family of difference sets
  is given by $(q, \frac{q-1}{2}, \frac{q-3}{4})$. Bent functions on
  the other hand give rise to difference sets in the elementary
  abelian group $G=\Z_2^n$. The connection is as follows: $D_f := \{x
  : f(x) = 1\}$ is a difference set in $\Z_2^n$ if and only if $f$ is
  a bent function, a result due to Dillon \cite{Dillon:75}. In this
  fashion we obtain $(2^n, 2^{n-1}\pm 2^{(n-2)/2}, 2^{n-2}\pm
  2^{(n-2)/2})$ difference sets in $\Z_2^n$, see also \cite{BJL1:99}.

\section{Quantum algorithms for the shifted bent function problem}

We introduce the hidden shift problem for Boolean functions. In
general, the hidden shift problem is a quite natural source of
problems for which a quantum computer might have an advantage over a
classical computer.  See \cite{CvD:2009} for more background on hidden
shifts and related problems. 

\begin{definition}[Hidden shift problem]
  Let $n \geq 1$ and let ${\cal O}_{f}$ be an oracle which gives
  access to two Boolean functions $f,g : \Z_2^n \rightarrow \Z_2$ such
  that the following conditions hold: (i) $f$, and $g$ are bent
  functions, and (ii) there exist $s\in \Z_2^n$ such that
  $g(x)=f(x+s)$ for all $x\in \Z_2^n$. We then say that ${\cal O}_{f}$
  hides an instance of a shifted bent function problem for the bent
  function $f$ and the hidden shift $s \in \Z_2^n$. If in addition to
  $f$ and $g$ the oracle also provides access to the dual bent
  function $\widetilde{f}$, then we use the notation ${\cal
    O}_{f,\widetilde{f}}$ to indicate this potentially more powerful
  oracle.
\end{definition}

\begin{theorem}\label{th:alg1}
  Let ${\cal O}_{f,\widetilde{f}}$ be an oracle that hides an instance
  of a shifted bent function problem for a function $f$ and hidden
  shift $s$ and provides access to the dual bent function
  $\widetilde{f}$. Then there exists a polynomial time quantum
  algorithm ${\cal A}_1$ that computes $s$ with zero error and makes
  two quantum queries to ${\cal O}_{f, \widetilde{f}}$.
\end{theorem}
\begin{proof}
  Let $f:\Z_2^n \rightarrow \Z_2$ be the bent function. We have oracle
  access to the shifted function $g(x)=f(x+s)$ via the oracle, i.\,e.,
  we can apply the map $\ket{x}\ket{0} \mapsto \ket{x}\ket{f(x+s)}$
  where $s \in \Z_2^n$ is the unknown string.  Recall that whenever we
  have a function implemented as $\ket{x}\ket{0}\mapsto
  \ket{x}\ket{f(x)}$, we can also compute $f$ into the phase as $U_f:
  \ket{x} \mapsto (-1)^{f(x)}\ket{x}$ by applying $f$ to a qubit
  initialized in $\frac{1}{\sqrt{2}}(\ket{0}-\ket{1})$.  The hidden
  shift problem is solved by the following algorithm ${\cal A}_1$: (i)
  Prepare the initial state $\ket{0}$, (ii) apply the Fourier
  transform $H_2^{\otimes n}$ to prepare an equal superposition
  $\frac{1}{\sqrt{2^n}}\sum_{x\in \Z_2^n} \ket{x}$ of all inputs,
  (iii) compute the shifted function $g$ into the phase to get
  $\frac{1}{\sqrt{2^n}}\sum_{x\in \Z_2^n} (-1)^{f(x+s)}\ket{x}$, (iv)
  Apply $H_2^{\otimes n}$ to get $\sum_w (-1)^{sw^t} \widehat{f}(w)
  \ket{w} = \frac{1}{\sqrt{2^n}} \sum_w (-1)^{sw^t}
  (-1)^{\widetilde{f}(w)} \ket{w}$, (v) compute the function $\ket{w}
  \mapsto (-1)^{\widetilde{f}(w)}\ket{w}$ into the phase resulting in
  $\frac{1}{\sqrt{2^n}}\sum_w (-1)^{sw^t} \ket{w}$, where we have used
  the fact that $f$ is a bent function, and (vi) finally apply another
  Hadamard transform $H_2^{\otimes n}$ to get the state $\ket{s}$ and
  measure $s$. From this description it is clear that we needed one
  query to $g$ and one query to $\widetilde{f}$ to solve the problem,
  that the algorithm is exact, and that the overall running time is
  given by $O(n)$ quantum operations. A quantum circuit implementing
  this algorithm is shown in Figure \ref{fig:algs}(a).
\end{proof}

Next, we consider the situation where the oracle defines a hidden shift
problem but does not provide access to the dual bent function. It 
turns out that in this case we can still extract the hidden shift
with a polynomial time quantum algorithm, however the number of 
queries increases from constant to linear. 

\begin{theorem}\label{th:alg2}
  Let ${\cal O}_{f}$ be an oracle that hides an instance of a shifted
  bent function problem for a function $f$ and hidden shift $s$. Then
  there exists a polynomial time quantum algorithm ${\cal A}_2$ that
  computes $s$ with constant probability of success and makes $O(n)$
  queries to ${\cal O}_{f}$.
\end{theorem}

\begin{proof}
First, note that as in Theorem \ref{th:alg1} we can assume that the
oracle computes the functions $f, g:\Z_2^n \rightarrow \Z_2$ into the
phase.  Furthermore, we can assume that the oracle can be applied
conditionally on a bit $b$, i.\,e., we can apply the map
$\Lambda_1(U_f): \ket{b}\ket{x} \mapsto \ket{b}\ket{x}$ if $b=0$ and
$\ket{b}\ket{x} \mapsto \ket{b}(-1)^{f(x)}\ket{x}$ if $b=1$. Indeed,
using a Fredkin gate {\sc Fred} (see \cite{NC:2000}) which specified
by $\ket{b}\ket{x}\ket{y} \mapsto \ket{b}\ket{x}\ket{y}$ if $b=0$ and
$\ket{b}\ket{x}\ket{y} \mapsto \ket{b}\ket{y}\ket{x}$ if $b=1$, it is
easy to implement $\Lambda_1(U_f)$ as follows: $(\Lambda_1(U_f)\otimes
\onemat_{2^n}) \ket{b}\ket{x}\ket{0} = (\mbox{{\sc Fred}} \circ
(\onemat_2 \otimes U_f \otimes \onemat_{2^n})\circ \mbox{{\sc Fred}})
\ket{b}\ket{x}\ket{0}$, up to a global phase.

We prove the theorem by reducing to an abelian hidden subgroup problem
in the group $\Z_2^{n+1}$. To do this, we use $f$ and $g$ to define
``quantum functions'', namely $F: x \mapsto \sum_{y\in \Z_2^n}
(-1)^{f(x+y)} \ket{y}$ and $G: x \mapsto \sum_{y\in \Z_2^n}
(-1)^{g(x+y)} \ket{y}$. Observe that due to the bentness of $f$ and
$g$, the two functions $F$ and $G$ are injective quantum functions,
i.\,e., they are injective complex valued functions that with respect
to some basis, which in general might be different from the
computational basis, become classical injective functions. Indeed,
this follows from the fact that all derivatives of a bent function are
balanced, see Section \ref{sec:MMclass}. Now, a well known connection
between the hidden shift problem for injective functions $f$, $g$ over
an abelian group $A$ and a hidden subgroup problem can be used
\cite{Kuperberg:2005,FIMSS:2003}. For this, the hidden subgroup
problem is defined with respect to the semidirect product $A \rtimes
\Z_2$ where the action is given by inversion in $A$. In our case we
have $A \rtimes \Z_2 \cong Z_2^{n+1}$ since the inversion action is
trivial over $\Z_2$.  The hiding function for the HSP over $Z_2^{n+1}$
is defined as $H(b,x)=F(x)$ if $b=0$ and $H(b,x)=G(x)$ if $b=1$. This
defines a hidden subgroup $\{(0,0), (1,s)\}$ of order $2$, knowledge
of which clearly implies that we know $s$. Once we have shown how to
implement the hiding function $H$, the algorithm will therefore be the
standard algorithm for the HSP: (i) Prepare the initial state
$\ket{0}$, (ii) apply the Fourier transform $H_2^{\otimes n}$ to
prepare an equal superposition $\frac{1}{\sqrt{2^{n+1}}}\sum_{b\in
  \Z_2, x\in \Z_2^n} \ket{x}$ of all inputs, (iii) compute the
function into the second register to get
$\frac{1}{\sqrt{2^{n+1}}}\sum_{b \in \Z_2, x\in \Z_2^n}
\ket{b,x}\ket{H(b,x)}$, (iv) Apply $H_2^{\otimes n+1}$ to the first
register, and (v) measure the first register. This leads to a
measurement result $a \in \Z_2^{n+1}$ that satisfies $(1,s) a^t = 0$.
Repeating steps (i)-(v) a total number of $O(n)$ times, we get a
constant probability to uniquely characterize $s$ from the measurement
data.  Hence, the algorithm needs $O(n)$ queries to $f$ and $g$ to
solve the problem and the overall running time is given by $O(n^2)$
quantum operations. The function $H(b,x)$ can be implemented in a
straightforward way using Hadamard transforms, controlled NOT
operations \cite{NC:2000}, and the controlled oracle calls
$\Lambda_1(U_f)$ mentioned above. A quantum circuit implementing one
iteration of this algorithm is shown in Figure \ref{fig:algs}(b).
\end{proof}

It is perhaps interesting to note that the ``probabilistic method'' of
directly implementing $\widetilde{f}$ via sampling of $f$ at a
polynomial number of inputs and using the Chernoff bound is not
sufficient for our purposes (see e.\,g., \cite{Mansour:94} for the
argument that $\sum_{i \in I} \chi_S(x_i) f(x_i)$ is exponentially
close to $\widetilde{f}$ for all $S$ for a sample set $I$ of
polynomial size). The issue is that for bent functions we would have
to distinguish exponentially small Fourier coefficients $\pm
1/\sqrt{2^n}$. We conjecture that in the worst case it takes an
exponential number of queries to $f$ in order to implement one query
to $\widetilde{f}$, but have no proof for this.

\newcommand{\goodgap}{%
\hspace{\subfigtopskip}\qquad\quad%
\hspace{\subfigbottomskip}
}

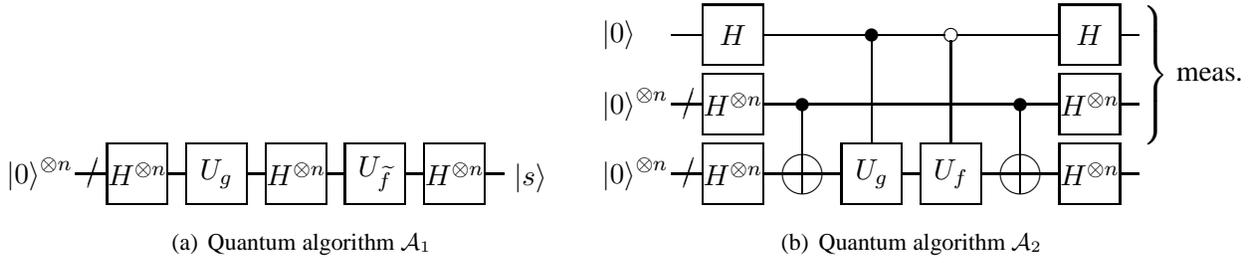
\begin{figure*}
\centering
\subfigure[Quantum algorithm ${\cal A}_1$]{
\unitlength0.75pt%
\hspace*{-0.5cm}
\begin{picture}(10,20)(0,-20)
\put(-25,0){\makebox(0,0)[l]{$\ket{0}^{\otimes n}$}}
\end{picture}%
\begin{picture}(16,20)(0,-20)
\put(0,0){\line(1,0){16}}
\put(6,-7){\line(1,3){5}}
\end{picture}%
\begin{picture}(30,20)(0,-20)
\put(0,-15){\framebox(30,30){$H^{\otimes n}$}}
\end{picture}%
\begin{picture}(10,20)(0,-20)
\put(0,0){\line(1,0){10}}
\end{picture}%
\begin{picture}(30,20)(0,-20)
\put(0,-15){\framebox(30,30){$U_g$}}
\end{picture}%
\begin{picture}(10,20)(0,-20)
\put(0,0){\line(1,0){10}}
\end{picture}%
\begin{picture}(30,20)(0,-20)
\put(0,-15){\framebox(30,30){$H^{\otimes n}$}}
\end{picture}%
\begin{picture}(10,20)(0,-20)
\put(0,0){\line(1,0){10}}
\end{picture}%
\begin{picture}(30,20)(0,-20)
\put(0,-15){\framebox(30,30){$U_{\widetilde{f}}$}}
\end{picture}%
\begin{picture}(10,20)(0,-20)
\put(0,0){\line(1,0){10}}
\end{picture}%
\begin{picture}(30,20)(0,-20)
\put(0,-15){\framebox(30,30){$H^{\otimes n}$}}
\end{picture}%
\begin{picture}(10,20)(0,-20)
\put(0,0){\line(1,0){10}}
\put(7,-7){\line(1,2){6}}
\end{picture}%
\begin{picture}(10,20)(0,-20)
\put(5,-2){\makebox(0,0)[l]{{$\ket{s}$}}}
\end{picture}%
}\goodgap
\subfigure[Quantum algorithm ${\cal A}_2$]{
\unitlength0.75pt%
\begin{picture}(10,80)(0,-20)
\put(-25,0){\makebox(0,0)[l]{$\ket{0}^{\otimes n}$}}
\put(-25,35){\makebox(0,0)[l]{$\ket{0}^{\otimes n}$}}
\put(-25,70){\makebox(0,0)[l]{$\ket{0}$}}
\end{picture}%
\begin{picture}(16,80)(0,-20)
\put(0,0){\line(1,0){16}}
\put(6,-7){\line(1,3){5}}
\put(0,35){\line(1,0){16}}
\put(6,28){\line(1,3){5}}
\put(0,70){\line(1,0){16}}
\end{picture}%
\begin{picture}(30,80)(0,-20)
\put(0,-15){\framebox(30,30){$H^{\otimes n}$}}
\put(0,20){\framebox(30,30){$H^{\otimes n}$}}
\put(0,55){\framebox(30,30){$H$}}
\end{picture}%
\begin{picture}(40,80)(0,-20)
\multiput(0,0)(0,35){3}{\line(1,0){40}}
\put(20,35){\circle*{6}}
\put(20,35){\line(0,-1){45}}
\put(20,0){\circle{20}}
\end{picture}%
\begin{picture}(30,80)(0,-20)
\put(0,-15){\framebox(30,30){$U_g$}}
\put(0,35){\line(1,0){30}}
\put(0,70){\line(1,0){30}}
\put(15,70){\circle*{6}}
\put(15,70){\line(0,-1){55}}
\end{picture}%
\begin{picture}(10,80)(0,-20)
\put(0,0){\line(1,0){10}}
\put(0,35){\line(1,0){10}}
\put(0,70){\line(1,0){10}}
\end{picture}%
\begin{picture}(30,80)(0,-20)
\put(0,-15){\framebox(30,30){$U_f$}}
\put(0,35){\line(1,0){30}}
\put(0,70){\line(1,0){12}}
\put(18,70){\line(1,0){12}}
\put(15,70){\circle{6}}
\put(15,67){\line(0,-1){52}}
\end{picture}%
\begin{picture}(40,80)(0,-20)
\multiput(0,0)(0,35){3}{\line(1,0){40}}
\put(20,35){\circle*{6}}
\put(20,35){\line(0,-1){45}}
\put(20,0){\circle{20}}
\end{picture}%
\begin{picture}(30,80)(0,-20)
\put(0,-15){\framebox(30,30){$H^{\otimes n}$}}
\put(0,20){\framebox(30,30){$H^{\otimes n}$}}
\put(0,55){\framebox(30,30){$H$}}
\end{picture}%
\begin{picture}(10,80)(0,-20)
\put(0,0){\line(1,0){10}}
\put(0,35){\line(1,0){10}}
\put(0,70){\line(1,0){10}}
\end{picture}%
\begin{picture}(20,80)(0,-20)
\put(0,50){\makebox(0,0)[l]{$\left. \rule{0mm}{1cm} \right\}$}}
\put(15,50){\makebox(0,0)[l]{{ meas.}}}
\end{picture}%
}
\caption{\label{fig:algs} Quantum algorithms for the hidden shift
  problem for bent functions. The quantum circuit in (a) implements
  algorithm ${\cal A}_1$. This algorithm can be used if access to the
  shifted function $g(x) = f(x+s)$ as well as access to the dual bent
  function $\widetilde{f}$ is given.  The algorithm uses one query to
  $g$ and one query to $\widetilde{f}$ and is zero-error, i.\,e., it
  always returns the hidden shift $s$. The quantum circuit in (b)
  implements algorithm ${\cal A}_2$. This algorithm uses access to $f$
  and $g$ only and can be applied if access to $\widetilde{f}$ is not
  available. The shown circuit has to be applied $O(n)$ times, after
  which the data acquired by measuring the upper $n+1$ qubits
  characterizes the hidden shift $s$ with constant probability of
  success.}
\end{figure*}

Finally, we state the two results that provide new query complexity
separations between quantum and classical algorithms. Our main tool is
the Maiorana-McFarland class of bent functions which turns out to be
rich enough to prove the two results. First, we show that the
classical query complexity for the hidden shift problem over this
class of bent functions is of order $\Theta(n)$, while it can be
solved with $2$ quantum queries.

\begin{theorem}\label{th:superpoly}
  Let ${\cal O}_{f,\widetilde{f}}$ be an oracle that hides a hidden
  shift $s$ for an instance $(f, g, \widetilde{f})$ of a hidden shift
  problem for a bent function $f$ from Maiorana-McFarland class. Then
  classically $\Theta(n)$ queries are necessary and sufficient to
  identify the hidden shift $s$. Further, there exists a recursively
  defined oracle ${\cal O}_{rec}$ which makes calls to ${\cal
    O}_{f,\widetilde{f}}$ and whose quantum query complexity is
  $poly(n)$, whereas its classical query complexity is
  superpolynomial.
\end{theorem} 

\begin{proof} The proof of the lower bound on the classical query
  complexity for ${\cal O}$ is information theoretic. The tightness of
  the bound follows since $n$ bits of information about $s$ have to be
  gathered and each query can yield at most $1$ bit. To see that
  $O(n)$ are indeed sufficient, consider the following (adaptive)
  strategy for finding a shift $(s,s^\prime)$ of $g(x,y) =
  (x+s)\pi(y+s^\prime)$: first query $g(x,y)$ on $(0,0)$ to extract
  $s\pi(s^\prime)$. Then subtract this from the values at the points
  $(e_i,0)$, where $e_i$ denotes the $i$th standard basis vector.
  This gives the bits of $\pi(s^\prime)$. Next evaluate
  $\widetilde{f}(x,y)=\pi^{-1}(x)y^t$ at the points $(\pi(s^\prime),
  e_i)$.  This gives the bits of $s^\prime$. Finally,
  from evaluating $g$ at points $(0,\pi^{-1}(e_i)+s^\prime)$ we can obtain
  the bits of $s$, i.\,e., the entire hidden shift $(s,s^\prime)$.

  A standard argument can be invoked \cite{BV:97} to recursively
  construct an oracle which hides a function computed by a tree, the
  nodes of which are given by the oracle hiding a string $s$. In order
  to evaluate $f(x)$ at a node, first a sequence of smaller instances
  of the problem have to be solved. We do not go into further detail
  of the construction and only note that we get the analogous result
  as in \cite{BV:97}, see also \cite{HH:2008}, namely that a tree of
  height $\log n$ leads to a quantum query complexity of $2^{\log n}$
  which is polynomial in $n$, whereas the classical query complexity
  is given by $n^{\log n}$ which grows faster than any polynomial.
\end{proof}

The following theorem avoids the adaptive queries in the proof of
Theorem \ref{th:superpoly} and uses oracles of the form ${\cal O}_f$
in which no queries to the dual bent function are allowed. Since the
quantum computer can still determine the shift in polynomial time,
here an exponential separation between classical and quantum query
complexity can be shown.

\begin{theorem}\label{th:exponential}
  Let ${\cal O}_{f}$ be an oracle that hides a hidden shift $s$ for an
  instance $(f, g)$ of a hidden shift problem for a bent function $f$
  from Maiorana-McFarland class. Then classically $\Theta(\sqrt{2^n})$
  queries are necessary and sufficient to identify the hidden shift
  $s$. 
\end{theorem} 

\begin{proof}
  The proof is similar to the lower bound for the linear structure
  problem considered in \cite{DCW:2002} and the query lower bound for
  Simon's problem \cite{Simon:94}. First, note that we can use Yao's
  minimax principle \cite{Yao:77} to show limitations of a
  deterministic algorithm ${\cal A}$ on the average over an
  adversarially chosen distribution of inputs.  Hence, we can consider
  deterministic algorithms and $\pi$ and $s$ in the definition of
  $f(x,y) = x \pi(y)^t$ and $g(x,y) = f(x,y+s)$ will be chosen randomly.

The distribution we chose to show the lower is to chose $\pi$
uniformly at random in $S_{2^n}$, the symmetric group on the strings
of length $n$, and $s=(s_1,s_2) \in \Z_2^{2n}$ such that $s_1=0$ and
$s_2$ is chosen uniform at random in $\Z_2^n$. The instances we
consider are given by oracle access to the functions $f(x,y) = x
\pi(y)^t$ and $g(x,y) = f(x, y+s) = x\pi(y+s)^t$. Now, without loss of
generality we can assume that the classical algorithm ${\cal A}$ has
(adaptively or not) queried the oracle $k=n^{O(1)}$ times, i.\,e., it
has chosen pairs $(x_i, y_i)$ for $i=1,\ldots,k$ and obtained results
\begin{eqnarray*}
x_i \pi(y_i)^t &=& a_i \\
x_i \pi(y_i+s)^t &=& b_i. 
\end{eqnarray*}

In order to characterize the information about $s$ after these $k$
queries we define set $D=\{x_i:i=1,\ldots,k\} \cup
\{y_i:i=1,\ldots,k\}$. We show that if no collision between the values
of $f$ and $g$ was produced, then the information obtained about $s$
is exponentially small. To simplify our argument, we actually make the
classical deterministic algorithm more powerful by giving oracle
access to $\pi(x)$ and $\pi(x+s)$. Consider the set of all differences
$D^{(-)} = \{d_1-d_2: d_1, d_2 \in D\}$ and the set $D_{good} =
\Z_2^n\setminus D^{(-)}$.  Note that for an abelian group $A$ and
subset $D\subset A$ with $|D|^2<|A|$ we can always choose a set $S$
such that $D \cap (D+s) = \emptyset$ for all $s\in S$. Indeed, we can
choose $S=D_{good}$ since $x \in D \cap (D+s)$ would imply that there
exist $d_1,d_2\in D$ with $d_1 = d_2+s$, i.\,e., $s\in D^{(-)}$ which
is a contradiction.  Notice in our case that $|S|\geq 2^n - |D^{(-)}|
= 2^n - n^{O(1)}$.  Now, we can change the value of the shift $s$ to
any other value $s^\prime$ as long as the algorithm has not queried
$s$ directly (the chances of which are exponentially small: because of
a birthday for the strings $s$, the probability is given by
$\Theta\left(\frac{1}{\sqrt{2^n}}\right)$). We do this by choosing
$\pi^\prime$ in such a way that it maps $\pi(y_i+s) =
\pi^\prime(y_i+s^\prime)$ while being consistent with all other
queries. Because of the above argument, as long as there is no
collision, after $\ell$ queries to $f$, $g$, we still have a set
$S$ of size $|S| \geq 2^n-n^{O(1)}$ of candidates $s^\prime$, and
$\pi^\prime$ which are also consistent with the sampled data, showing
the lower bound.
\end{proof}

\begin{corollary}
There exists an oracle ${\cal O}$ implementing a Boolean function 
such that ${\sf P}^{\cal O} \not= {\sf BQP}^{\cal O}$. 
\end{corollary}

\section{Conclusions}

We introduced the hidden shift problem for a class of Boolean
functions which are at maximum distance to all linear functions. For
these so-called bent functions the hidden shift problem can be
efficiently solved on a quantum computer, provided that we have oracle
access to the shifted version of the function as well as its dual bent
function.  The quantum computer can extract the hidden shift using
just one query to these two functions and besides this only requires
to compute the Hadamard transform and measure qubits in the standard
basis.  We showed that this task is significantly more challenging for
a classical computer and proved an exponential separation between 
quantum and classical query complexity. 


\begin{thebibliography}{HMR{\etalchar{+}}06}

\bibitem[Aar03]{Aaronson:2003}
S.~Aaronson.
\newblock {Quantum lower bound for recursive Fourier sampling}.
\newblock {\em Quantum Information and Computation}, 3(2):165--174, 2003.

\bibitem[AS07]{AS:2007}
A.~Atici and R.~Servedio.
\newblock Quantum algorithms for learning and testing juntas.
\newblock {\em Quantum Information Processing}, 6(5):323--348, 2007.

\bibitem[BCvD05]{BCvD:2005}
D.~Bacon, A.~Childs, and {W. van} {D}am.
\newblock From optimal measurement to efficient quantum algorithms for the
  hidden subgroup problem over semidirect product groups.
\newblock In {\em Proceedings of the 46th Annual IEEE Symposium on Foundations
  of Computer Science}, pages 469--478, 2005.

\bibitem[dBCW02]{DCW:2002}
N.~de~Beaudrap, R.~Cleve, and J.~Watrous.
\newblock Sharp quantum versus classical query complexity separations.
\newblock {\em Algorithmica}, 34(4):449--461, 2002.

\bibitem[BJL99]{BJL1:99}
Th. Beth, D.~Jungnickel, and H.~Lenz.
\newblock {\em Design Theory}, volume~I.
\newblock Cambridge University Press, 2nd edition, 1999.

\bibitem[BV97]{BV:97}
E.~Bernstein and U.~Vazirani.
\newblock Quantum complexity theory.
\newblock {\em SIAM Journal on Computing}, 26(5):1411--1473, 1997.
\newblock Conference version in Proc. STOC'93, pp.~11--20.

\bibitem[CG06]{CG:2006}
C.~Carlet and Ph. Gaborit.
\newblock Hyper-bent functions and cyclic codes.
\newblock {\em Journal of Combinatorial Theory, Ser. A}, 113:466--482, 2006.

\bibitem[CvD07]{CvD:2007}
A.~Childs and {W. van} Dam.
\newblock Quantum algorithm for a generalized hidden shift problem.
\newblock In {\em {Proceedings of the 18th Symposium on Discrete Algorithms
  (SODA'07)}}, pages 1225--1232, 2007.

\bibitem[CvD08]{CvD:2009}
A.~Childs and {W. van} Dam.
\newblock Quantum algorithms for algebraic problems.
\newblock arXiv Preprint 0812.0380, to appear in Reviews of Modern Physics.

\bibitem[CSV07]{CSV:2007}
A.~Childs, L.~J. Schulman, and U.~Vazirani.
\newblock Quantum algorithms for hidden nonlinear structures.
\newblock In {\em Proceedings of the 48th Annual IEEE Symposium on Foundations
  of Computer Science (FOCS'07)}, pages 395--404, 2007.

\bibitem[CW07]{CW:2007}
A.~Childs and P.~Wocjan.
\newblock {On the quantum hardness of solving isomorphism problems as
  nonabelian hidden shift problems}.
\newblock {\em Quantum Information and Computation}, 7(5--6):504--521, 2007.

\bibitem[vDHI03]{vDHI:2003}
{W. van} Dam, S.~Hallgren, and L.~Ip.
\newblock Quantum algorithms for some hidden shift problems.
\newblock In {\em {Proceedings of the 14th Symposium on Discrete Algorithms
  (SODA'03)}}, pages 489--498, 2003.

\bibitem[DDW08]{DDW:2008}
Th. Decker, J.~Draisma, and P.~Wocjan.
\newblock {Efficient quantum algorithm for identifying hidden polynomials}.
\newblock {\em Quantum Information and Computation}, 2008.
\newblock To appear, see also arxiv preprint 0706.1219".

\bibitem[DJ92]{DJ:92}
D.~Deutsch and R.~Jozsa.
\newblock Rapid solution of problems by quantum computation.
\newblock {\em Proceedings of the Royal Society London, Series A},
  439:553--558, 1992.

\bibitem[Dil75]{Dillon:75}
J.~Dillon.
\newblock {Elementary Hadamard difference sets}.
\newblock In F.~(et~al.) Hoffman, editor, {\em {Proc. 6th S-E Conf. on
  Combinatorics, Graph Theory, and Computing}}, pages 237--249. Winnipeg
  Utilitas Math., 1975.

\bibitem[Dob95]{Dobbertin:95}
H.~Dobbertin.
\newblock {Construction of bent functions and balanced Boolean functions with
  high nonlinearity}.
\newblock In B.~Preneel, editor, {\em {Fast Software Encryption}}, volume 1008
  of {\em LNCS}, Springer, pages 61--74, 1995.

\bibitem[FIM{\etalchar{+}}03]{FIMSS:2003}
K.~Friedl, G.~Ivanyos, F.~Magniez, M.~Santha, and P.~Sen.
\newblock {Hidden translation and orbit coset in quantum computing}.
\newblock In {\em Proc. STOC'03}, pages 1--9, 2003.

\bibitem[Hal02]{Hallgren:2002}
S.~Hallgren.
\newblock {Polynomial-time quantum algorithms for Pell's equation and the
  principal ideal problem}.
\newblock In {\em Proc. STOC'02}, pages 653--658, 2002.

\bibitem[HH08]{HH:2008}
S.~Hallgren and A.~Harrow.
\newblock {Superpolynomial speedups based on almost any quantum circuit}.
\newblock In {\em {Proceedings of the 35th International Colloquium on 
Automata, Languages and Programming (ICALP'08)}}, pages 782--795, 2008.

\bibitem[HMR{\etalchar{+}}06]{HMRRS:2006}
S.~Hallgren, C.~Moore, M.~R{\"o}tteler, A.~Russell, and P.~Sen.
\newblock Limitations of quantum coset states for graph isomorphism.
\newblock In {\em Proceedings of the 38th Annual ACM Symposium on Theory of
  Computing (STOC'06)}, pages 604--617, 2006.

\bibitem[HZ99]{HW:99}
T.~Helleseth and V.~Zinoviev.
\newblock {On $Z_4$-linear Goethals codes and Kloosterman sums}.
\newblock {\em Designs, Codes and Cryptography}, 17:269--288, 1999.

\bibitem[Kit97]{Kitaev:97}
A.~Yu. Kitaev.
\newblock Quantum computations: algorithms and error correction.
\newblock {\em Russian Math. Surveys}, 52(6):1191--1249, 1997.

\bibitem[KS07]{KS:2007}
A.~R. Klivans and A.~A. Sherstov.
\newblock {Unconditional lower bounds for learning intersections of
  halfspaces}.
\newblock {\em Machine Learning}, 69(2--3):97--114, 2007.

\bibitem[Kup05]{Kuperberg:2005}
G.~Kuperberg.
\newblock {A subexponential-time quantum algorithm for the dihedral hidden
  subgroup problem}.
\newblock {\em SIAM Journal on Computing}, 35(1):170--188, 2005.

\bibitem[LW90]{LW:90}
G.~Lachaud and J.~Wolfmann.
\newblock {The weights of the orthogonals of the extended quadratic binary
  Goppa codes}.
\newblock {\em IEEE Transactions on Information Theory}, 36(3):686--692, 1990.

\bibitem[Man94]{Mansour:94}
Y~Mansour.
\newblock {Learning Boolean functions {\em via} the Fourier transform}.
\newblock In V.~P. Roychodhury, K.-Y. Siu, and A.~Orlitsky, eds., {\em
  Theoretical Advances in Neural Computation and Learning}, pp. 391--424.
  Kluwer, 1994.

\bibitem[LMS08]{LMS:2008}
S.~Lovett, R.~Meshulam, and A.~Samorodnitsky.
\newblock {Inverse conjecture for the Gowers norm is false}.
\newblock In {\em Proceedings of the 40th Annual ACM Symposium on Theory of
  Computing (STOC'08)}, pages 547--556, 2008.

\bibitem[MS77]{MS:77}
F.~J. MacWilliams and N.~J.~A. Sloane.
\newblock {\em The Theory of Error--Correcting Codes}.
\newblock North--Holland, Amsterdam, 1977.

\bibitem[Mon09]{Montanaro:2009}
A.~Montanaro.
\newblock Quantum algorithms for shifted subset problems.
\newblock {\em Quantum Information and Computation}, 9(5\&6):500--512, 2009.

\bibitem[MRRS07]{MRR+:2007}
C.~Moore, D.~N. Rockmore, A.~Russell, and L.~J. Schulman.
\newblock {The power of strong Fourier sampling: quantum algorithms for affine
  groups and hidden shifts}.
\newblock {\em SIAM Journal on Computing}, 37(3):938--958, 2007.

\bibitem[NC00]{NC:2000}
M.~Nielsen and I.~Chuang.
\newblock {\em Quantum Computation and Quantum Information}.
\newblock Cambridge University Press, 2000.

\bibitem[Rot76]{Rothaus:76}
O.~S. Rothaus.
\newblock On ``bent'' functions.
\newblock {\em Journal of Combinatorial Theory, Series A}, 20:300--305, 1976.

\bibitem[RS04]{RS:2004}
A.~Russell and I.~Shparlinski.
\newblock {Classical and quantum function reconstruction via character
  evaluation}.
\newblock {\em Journal of Complexity}, 20(2--3):404--422, 2004.

\bibitem[Sho97]{Shor:97}
P.~Shor.
\newblock Polynomial-time algorithms for prime factorization and discrete
  logarithms on a quantum computer.
\newblock {\em SIAM Journal on Computing}, 26(5):1484--1509, 1997.

\bibitem[Sim94]{Simon:94}
D.~R. Simon.
\newblock On the power of quantum computation.
\newblock In {\em Proceedings of the 35th Annual Symposium on Foundations of
  Computer Science (FOCS'94)}, pages 116--123, 1994.

\bibitem[dW08]{DeWolf:2008}
R.~de~Wolf.
\newblock {A brief introduction to Fourier analysis on the Boolean cube}.
\newblock {\em Theory of Computing Library--Graduate Surveys}, 1:1--20, 2008.
\newblock Available online from www.theoryofcomputing.org.

\bibitem[Yao77]{Yao:77}
A.~Yao.
\newblock Probabilistic computations: toward a unified measure of complexity.
\newblock In {\em Proceedings of the 18th Annual Symposium on Foundations of
  Computer Science (FOCS'77)}, pages 222--227, 1977.

\end{thebibliography}

\section*{Acknowledgments} 
The author gratefully acknowledges support by ARO/NSA under grant
W911NF-09-1-0569 and would like thank the anonymous referees for
valuable comments on this paper and earlier versions of it. 

\newcommand{\etalchar}[1]{$^{#1}$}

\end{document}